\pdfoutput=1
\documentclass[10pt]{amsart}
\usepackage[utf8]{inputenc}
\usepackage{bm}
\def\ALGS{0}

\usepackage{algorithm}
\usepackage[noend]{algpseudocode}
\algnewcommand{\algorithmicassumption}{\textbf{Requirement:}}
\algnewcommand{\Assume}{\item[\algorithmicassumption]}
\algrenewcomment[1]{\hfill \(\triangleright\) \emph{\small #1}} 
\algnewcommand{\InlineIf}[2]{
  \algorithmicif\ #1\ \algorithmicthen\ #2}
\algnewcommand{\InlineElse}[1]{
  \algorithmicelse\ #1}
\algnewcommand{\InlineIfElse}[3]{
  \algorithmicif\ #1\ \algorithmicthen\ #2\ \algorithmicelse\ #3}
\algnewcommand{\InlineFor}[2]{\algorithmicfor\ #1\ \algorithmicdo\ #2} 
\algnewcommand{\CommentLine}[1]{\(\triangleright\) \emph{\small #1}}

\algnewcommand{\algorithmicand}{\textbf{and}}
\algnewcommand{\algorithmicor}{\textbf{or}}
\algnewcommand{\FOR}{\algorithmicfor}
\algnewcommand{\OR}{\algorithmicor}
\algnewcommand{\AND}{\algorithmicand}
\algnewcommand{\IF}{\algorithmicif}
\algnewcommand{\THEN}{\algorithmicthen}
\algnewcommand{\ELSE}{\algorithmicelse}
\usepackage{amsmath}
\usepackage{amssymb}
\usepackage{amsthm}
\usepackage{mathtools}
\usepackage{mathrsfs}
\usepackage{amsfonts}
\usepackage[shortlabels]{enumitem}
\usepackage{hyperref}
\usepackage[capitalise]{cleveref}
\crefname{problem}{problem}{problems}
\Crefname{problem}{Problem}{Problems}
\usepackage{multirow}

\theoremstyle{definition}
\newtheorem{theorem}{Theorem}
\newtheorem{proposition}[theorem]{Proposition}

\theoremstyle{remark}
\newtheorem{remark}[theorem]{Remark}

\newcommand{\ZZ}{\mathbb{Z}}
\newcommand{\ZZp}{\mathbb{Z}_{> 0}}  

\newcommand{\QQ}{\mathbb{Q}}
\newcommand{\FF}{\mathbb{F}}

\renewcommand{\AA}{\mathbb{A}}

\newcommand{\Syz}{\operatorname{Syz}}
\newcommand{\trace}{\operatorname{tr}}

\newcommand{\HF}{\operatorname{HF}}
\newcommand{\rk}{\operatorname{rk}}
\newcommand{\rows}{\operatorname{rows}}

\newcommand{\im}{\operatorname{im}}

\newcommand{\fraka}{\mathfrak{a}}

\newcommand{\field}{\Bbbk}  
\newcommand{\closure}{\bar{\field}}  
\newcommand{\pring}{\field[x_1,\ldots,x_k]} 
\newcommand{\ring}{\mathcal{R}} 
\newcommand{\ideal}{\mathcal{I}} 
\newcommand{\module}{\mathcal{M}} 
\newcommand{\matmod}{\module_n(\ring)} 
\newcommand{\emodule}{\mathcal{E}} 
\newcommand{\macmat}{\mathscr{M}} 
\newcommand{\redmacmat}{\bar{\mathscr{M}}} 
\newcommand{\genby}[1]{\langle #1 \rangle} 
\newcommand{\genbyPar}[1]{\left\langle #1 \right\rangle} 
\newcommand{\detidealGen}[2]{\ideal_{#2}(#1)} 
\newcommand{\detideal}{\detidealGen{M}{r}} 
\newcommand{\detidealCorkOne}{\detidealGen{M}{n-2}} 
\newcommand{\detsystemGen}[2]{F_{#2}(#1)} 
\newcommand{\detsystem}{\detsystemGen{M}{r}} 
\newcommand{\detsystemCorkOne}{\detsystemGen{M}{n-2}} 
\newcommand{\ordsig}{\prec_{\mathrm{sig}}} 
\newcommand{\ord}{\prec} 
\newcommand{\ordpot}{\prec_{\mathrm{pot}}} 
\newcommand{\lt}{\operatorname{lt}} 
\newcommand{\ltpot}{\operatorname{lt}_{\mathrm{pot}}} 
\newcommand{\cofac}[1]{{#1}^*} 
\newcommand{\coeffa}{\boldsymbol{a}} 
\newcommand{\mathtrue}{\mathtt{true}} 
\newcommand{\mathfalse}{\mathtt{false}} 
\newcommand{\property}{\mathscr{P}} 
\newcommand{\genmat}{\mathscr{A}}
\newcommand{\gendetideal}{\detidealGen{\genmat}{r}}
\newcommand{\gendetsystem}{\detsystemGen{\genmat}{r}}
\newcommand{\gendetidealCorkOne}{\detidealGen{\genmat}{n-2}}
\newcommand{\CM}{\mathrm{CM}} 
\newcommand{\matringbasis}[2]{\boldsymbol{E}_{#1,#2}}
\newcommand{\matentry}[2]{m_{#1,#2}}
\newcommand{\affspaceCorkOne}{\AA^{4\cdot n^2}_{\field}} 
\newcommand{\affspace}{\AA^{k n^2}_{\field}} 
\newcommand{\myparagraph}[1]{\smallskip\emph{#1.}} 

\usepackage[foot]{amsaddr}
\author{Sriram Gopalakrishnan}
\author{Vincent Neiger}
\author{Mohab Safey El Din}
\address{Authors' affiliation: Sorbonne Université, CNRS, LIP6, F-75005 Paris, France}
\email{firstname.lastname@lip6.fr}

\title{Refined \texorpdfstring{F\textsubscript{5}}{F5} Algorithms for Ideals of Minors of Square Matrices}

\begin{document}

\maketitle

\begin{abstract}
We consider the problem of computing a grevlex Gröbner basis for the set
$F_r(M)$ of minors of size $r$ of an $n\times n$ matrix $M$ of generic linear
forms over a field of characteristic zero or large enough. Such sets are
not regular sequences; in fact, the ideal \(\langle F_r(M) \rangle\) cannot be
generated by a regular sequence. As such, when using the general-purpose
algorithm $F_5$ to find the sought Gröbner basis, some computing time
is wasted on reductions to zero. We use known results about the first syzygy
module of $F_r(M)$ to refine the $F_5$ algorithm in order to detect more
reductions to zero. In practice, our approach avoids a significant number of
reductions to zero. In particular, in the case $r=n-2$, we prove that our new
algorithm avoids all reductions to zero, and we provide a corresponding
complexity analysis which improves upon the previously known estimates.
\end{abstract}

\section{Introduction}
\label{sec:intro}

\myparagraph{Motivation and problem}
Let $M$ be an $n\times n$ matrix with entries in the polynomial ring \(\ring =
\pring\) where $\field$ is a field. For $r < n$, we let $\detideal$ be the
determinantal ideal generated by the sequence $\detsystem$ of all minors of
\(M\) of size $r+1$. We consider the problem of computing the common roots in $\closure^{k}$ to
$\detsystem$, hence those points at which $M$ has rank at most $r$.
This $\mathcal{NP}$-hard problem \emph{MinRank}~\cite{BussFrandsenShallit1999},
and its variants where $M$ may be rectangular, lies at the heart of
multivariate cryptography. It is at the foundations of
several schemes~\cite{Courtois2001,Patarin1996,KipnisShamir1999} and is still
used to assess the security of encryption and signature
schemes~\cite{FaugereLevyPerret2008,DingSchmidt2005,Beullens2022,
BaenaBriaudCabarcasPerlnerSmithToneVerbel2022,
BardetBriaudBrosGaboritNeigerRuattaTillich2020,
BardetBrosCabarcasGaboritPerlnerSmithToneTillitchVerbel2020}.

Determinantal ideals also arise in fundamental areas such as effective real
algebraic geometry as they encode critical points (see e.g.~\cite{FSS12, Spa14}),
then used to solve a variety of problems. This includes polynomial
optimization~\cite{GS14,BGHM14}, computing sample points and answering connectivity queries in smooth real algebraic
sets~\cite{SaSc03, BGHP05, BGHSS, SaSc17}, determining the
dimension of real algebraic sets~\cite{BaSa15, LaSa21}, and quantifier elimination over the
reals~\cite{HoSa09,HoSa12,LeSa21}.

\myparagraph{Determinantal ideals and polynomial system solving}
Determinantal ideals enjoy plenty of combinatorial and algebraic
properties~\cite{BrunsVetter1988, Lascoux1978, BCRV22} which can be leveraged to
better understand the complexity of computing their roots, and to adapt and
accelerate polynomial system solvers in this context. The most advanced results
in
this direction have been achieved in the context of symbolic homotopy techniques
with the design of an adapted homotopy pattern~\cite{hauenstein:hal-01719170}
which has next been refined to take into account specific structures when the
entries of the matrix $M$ are sparse~\cite{LSSV21}.

In this paper, we focus on the problem of computing {\em Gröbner bases} of the
ideal $\detideal$ w.r.t.\ some admissible monomial ordering, under the assumptions
that $\detideal$ has dimension $0$ (or is \(\ring\)) and that the entries of
$M$ have total degree at most $1$.

\myparagraph{Gröbner bases algorithms and determinantal ideals}
Since Buchberger's algorithm~\cite{bGroebner1965}, the quest for fast algorithms
for computing Gröbner bases has been driven by two main issues: {\em (i)} finding
better strategies for handling critical pairs during the Gröbner basis
construction and {\em (ii)} hunting reductions to $0$ which are
instrinsically related to algebraic objects named syzygies that are
associated to the ideal under consideration. Issue {\em (i)} has been addressed
by Faugère's celebrated $F_4$ algorithm \cite{Faugere1999}, which also made
explicit the use of linear algebra subroutines in Gröbner bases algorithms. While a
lot remains to be done in this direction (see e.g. \cite{berthomieu:hal-03590430}), much attention
has focused on issue {\em (ii)} and variants of Faugère's $F_5$
algorithm~\cite{Faugere2002} have been developed in several directions to give
rise to signature-based Gröbner bases algorithms (see~\cite{EderFaugere2016} and
references therein). One byproduct of these works, which finds its roots in
foundational works by Lazard and Giusti~\cite{Lazard1983,Giusti1984}, is that
they paved the way to complexity estimates {\it under some regularity
assumptions}, thanks to the reduction to linear algebra and degree bounds on the
maximum degree reached during the computation (related to the classical notion
of index of regularity~\cite[Chap.\,9,\,\S3]{CoxLittleOShea2007}).

This has been developed, for determinantal ideals in~\cite{FSS10,
  FaugereSafeySpaenlehauer2013} which yield complexity estimates for computing
Gröbner bases {\em under regularity assumptions} (which are generic in the sense
of algebraic geometry). These estimates are coarse: they do not leverage the
shape of the matrices encountered during the computation.

Already in the simpler case of regular sequences, 
by exploiting the
fact that the $F_5$ algorithm avoids all reductions to zero in this case,
a sharper complexity analysis
of $F_5$ \cite{BardetFaugereSalvy2015} shows significant improvements
against such coarse estimates.

In the context of determinantal ideals, mimicking this to get better complexity
estimates is premature. Indeed, $\detsystem$ is {\em not} a regular
sequence, and running the $F_5$ algorithm with input $\detsystem$ does lead to a
number of reductions to $0$. Hence there is a need to refine and tune the $F_5$
algorithm for determinantal ideals. Such a refinement has already been achieved
for boolean polynomial systems~\cite{BardetFaugereSalvySpaenlehauer2013}.
However, recall that these reductions to $0$ are related to so-called syzygy
modules of the ideal under study. Syzygy modules of determinantal ideals are
notoriously more intricate than those of ideals generated by regular sequences
or boolean systems. 

In this paper, we tackle the following problems: {\em (i)} What is the suitable
notion of regularity one can attach to determinantal ideals in order to hunt
reductions to $0$? {\it (ii)} What are the properties of modules of syzygies
associated to determinantal ideals one can leverage under this notion of
regularity? {\em (iii)} How to refine the $F_5$ algorithm for determinantal ideals
to obtain fewer reductions to $0$ and, ultimately, are there some instances
of determinantal ideals for which one can prove that there are no reductions to
$0$?

\myparagraph{Foundations}
We begin by recalling first the connection between free resolutions and syzygy
modules of ideals, and then the \textit{syzygy criterion} from
\cite{EderFaugere2016} which reveals the link between free resolutions and
reductions to zero in $F_5$. In \cref{alg:mF5-syz}, we give an altered version
of the standard matrix-$F_5$ algorithm: it computes Gröbner bases for modules
over $\ring$ and exploits the full syzygy criterion (see \cref{prop:syz-crit}),
allowing us to leverage reductions to zero in lower degrees to avoid reductions
to zero in subsequent degrees.

Explicitly describing free resolutions of determinantal ideals is
  in general an extremely difficult problem. It is in fact not known if there
  even exists a minimal free resolution of the system of $(r+1)$-minors of a
  matrix of indeterminates over $\ZZ$ which remains minimal under arbitrary base
  change. While the Lascoux resolution \cite{Lascoux1978} provides a free
  resolution of determinantal ideals, it is not minimal and requires that the
  coefficient ring have characteristic zero. Instead of computing a free
  resolution of these determinantal ideals directly, we instead adopt a strategy
  which relies on a theorem of Kurano \cite{kurano1989}. It describes
  the connection between syzygies of $(r+1)$-minors of a matrix $M$ and syzygies
  between $(r+1)$-minors of the $(r+2)\times(r+2)$ submatrices of $M$. Using
  this theorem, we essentially reduce to the case $r=n-2$. In this case, an
  explicit free resolution exists, given by Gulliksen and Neg{\aa}rd in
  \cite{BrunsVetter1988}.

\myparagraph{Main results}
Having made this reduction, we establish the genericity property which our
ideals must satisfy in order for the Gullik\-sen-Neg{\aa}rd complex to be exact:
for any $1\le r < n$, the ideal of $(r+1)$-minors of an $n\times n$ matrix of
indeterminates has the so-called \emph{Cohen-Macaulay} property. Thus, for a
suitably generic choice of coefficients of the linear forms in $M$, the ideal
$\detideal$ is Cohen-Macaulay as well. It is precisely under the genericity
assumption derived from this notion that the complex of Gulliksen and Neg{\aa}rd
is a free resolution of $\detidealCorkOne$, and can therefore be exploited to
avoid reductions to zero.

By tracing basis elements for the free modules which make up the complex of
Gulliksen and Neg{\aa}rd, we are able (\cref{thm:syz-corank-1}) to explicitly
compute a generating set for the first syzygy module of the system of
$(n-1)$-minors of an $n\times n$ matrix of linear forms, provided the above
stated genericity assumption holds. Kurano' s result \cite{kurano1989} states
that for any $1\le r < n$, the first module of syzygies $\Syz(\detsystem)$ is
generated by the syzygies between the $(r+1)$-minors of each $(r+2)\times (r+2)$
submatrix of $M$.

Therefore, combining the complex of Gulliksen and Neg{\aa}rd with the result of
\cite{kurano1989}, we are able to explicitly compute a full generating set for
$\Syz(\detsystem)$, and subsequently provide
\cref{alg:det-F5}, which computes a grevlex Gröbner basis for $\detideal$ while
avoiding all reductions to zero which arise from the syzygies in degree one.

Under our genericity assumption, when $r=n-2$, the Gulliksen-Neg{\aa}rd complex
allows us to compute generating sets for the higher syzygy modules of
$\detsystemCorkOne$ as well. In \cref{prop:syz-2-corank-1}, we give explicit
generators for the second syzygy module of $\detidealCorkOne$. This study
culminates in \cref{alg:det-F5-corank-one} which is an altered version of
matrix-$F_5$ which avoids all reductions to zero. Finally, in
\cref{prop:det-hilb-series}, we again exploit the Gulliksen-Neg{\aa}rd complex
to provide an explicit form for the Hilbert series of $\detidealCorkOne$ when
the entries of $M$ are sufficiently generic homogeneous linear forms, and when $\detidealCorkOne$ has dimension zero ($k=4$). In
\cref{prop:corank-one-complexity}, we use this series to give a complexity
analysis of our new algorithm in the case $r=n-2$, demonstrating that
asymptotically, the arithmetic complexity of our new algorithm is in
$O(n^{4\omega-1})$, while the current best-known asymptotic
arithmetic complexity of computing a grevlex Gröbner basis for
$\detidealCorkOne$ is in $O(n^{5\omega+2})$. Here,
$2\le\omega\le 3$ is a complexity exponent for matrix multiplication. 

We conclude by giving, in \cref{table:corank-one}, some experimental data
showing the amount of reductions to zero that is saved by our contributions and
their practical interest.

\myparagraph{Perspectives}
In \cite{ma1994}, it is shown that in
some cases, one can obtain generators for the second syzygy module of
$\detideal$ by lifting second syzygies of minors of submatrices, as is the case
for first syzygies. Thus, the careful treatment of the Gulliksen-Neg{\aa}rd
complex which we give in this paper could be exploited in future works to avoid
more reductions to zero when $r<n-2$.

Similarly, suppose $M$ is no longer a square matrix, but is instead an $n\times
m$, $n\ne m$ matrix of generic homogeneous linear forms over $\field$. Then when
$r=\min(n,m)-1$ so that $\detideal$ is the ideal of maximal minors of $M$, the
Eagon-Northcott complex (see \cite[2.C]{BrunsVetter1988} and
\cite{EagonNorthcott1962}) provides a free resolution of $\detideal$.
Similarly, when $r=\min(n,m)-2$, the Akin-Buschbaum-Weyman complex (see
\cite{AkinBuschbaumWeyman1981}) provides a free resolution of $\detideal$.
Again, the tools and methods brought in this paper could be adapted to accelerate
Gröbner bases computations in this case and yield new complexity bounds. 

Finally, in full generality, the Lascoux resolution (see \cite{Lascoux1978}),
is a free resolution for $\detideal$ for any $n,m,r$ provided
$\QQ\subseteq\field$. Again, one may expect refined F\textsubscript{5}
algorithms by leveraging this resolution.

\section{Preliminaries}
\label{sec:prelim}

\subsection{Syzygies}
\label{sec:prelim:syzygies}

We recall basic definitions and properties of syzygy modules, when working over
the Noetherian ring \(\ring = \pring\). We refer to \cite{Eisenbud1995}
for more details. For a finitely generated \(\ring\)-module $\module =
\genby{p_1,\dots, p_\ell}$, the \emph{first syzygy module} of $\module$ is
defined as
\[
	\Syz(\module) :=
  \{(s_1,\ldots,s_\ell) \in \ring^{\ell} :
  s_1 p_1 + \cdots + s_\ell p_\ell = 0 \} .
\] 
This definition depends on the generators; we sometimes write
\(\Syz(p_1,\ldots,p_\ell)\). From there one inductively defines the \emph{$j$-th
  syzygy module of \(\module\)} as follows. Since $\ring$ is Noetherian,
$\Syz_{j-1}(\module)$ is finitely generated. With generators $\{q_1,\dots,q_t\}$
for $\Syz_{j-1}(\module)$,
\[
	\Syz_j(\module) :=
  \{(s_1,\dots, s_t)\in \ring^{t} :
  s_1 q_1 + \cdots + s_t q_t = 0 \} .
\] 

It is frequent that \(\module\) is the ideal generated by polynomials $F =
(f_1,\dots, f_\ell)  \subseteq\ring$. Then, the first syzygy module of \(F\)
contains the Koszul syzygies, which are those following from the commutativity
of polynomial multiplication: $f_if_j-f_jf_i=0$. In fact, they generate
\(\Syz(F)\) in the case of \emph{regular sequences} (that is, when $f_i$ is not
a zero-divisor in $\ring/\genby{f_1,\dots, f_{i-1}}$ for any $2\le i\le\ell$):

\begin{theorem}[{\cite[Thm.\,A.2.49]{Eisenbud2005}}]
  \label{thm:syz-regular-seq}
  If $(f_1,\dots,f_\ell)$ is a regular sequence, then $\Syz(F) =
  \genbyPar{f_ie_j-f_je_i:1\le i,j\le\ell, i\ne j}$ where $e_i$ is $i$-th
  standard basis vector. 
\end{theorem}

In the context of $\detsystem$, while the Koszul syzygies are among the syzygies
of the minors of $M$, they do not generate $\Syz(\detsystem)$.

\subsection{Free resolutions}
\label{sec:prelim:resolutions}

As highlighted in \cref{sec:intro}, in relation to the \(k\)-th syzygy module of
\(\detsystem\), our approach involves the description of a \emph{free resolution}
of \(\detideal\) (when \(r=n-2\)). For a finitely generated \(\ring\)-module
$\module$, a free resolution of $\module$ is an exact complex
\[
	\cdots
  \xrightarrow{d_{t+1}} \emodule_{t}
  \xrightarrow{d_t} \emodule_{t-1}
  \xrightarrow{d_{t-1}} \cdots
  \xrightarrow{d_2} \emodule_1
  \xrightarrow{d_1}  \emodule_0
  \xrightarrow{\epsilon} \module
  \to 0
\] 
where for each $j>0$, $\emodule_j$ is a finitely generated free $\ring$-module,
and the $d_j$ are $\ring$-module homomorphisms. The exactness condition
precisely means that
$\ker(d_j)=\im(d_{j+1})$. 
The free resolution $\emodule_\bullet$ is said to be \emph{finite} if there
exists some $m \ge 0$ such that for all $j>m$, $\emodule_j=\{0\}$; then
the smallest such $m$ is called the \emph{length} of
$\emodule_\bullet$.  In general, modules need not have finite free resolutions;
however, it is the case for finitely generated modules over \(\ring=\pring\):

\begin{theorem}[Hilbert's syzygy theorem]
  \label{thm:hilbert-syzygy}
  Let $\module$ be a finitely generated \(\ring\)-module. There exists a
  free resolution
	\[
		0\to \emodule_m\xrightarrow{d_m}\emodule_{m-1}\xrightarrow{d_{m-1}}\cdots\xrightarrow{d_2}\emodule_1\xrightarrow{d_1} \emodule_0\xrightarrow{\epsilon} \module\to 0
	\] 
  whose length \(m\) is at most the number of variables \(k\).
\end{theorem}

\begin{proposition}
  \label{prop:syzygy-resolution}
  Let $\module$ be a finitely generated \(\ring\)-module, $\emodule_\bullet$ be
  a free resolution of $\module$ of length $m\le k$, and $\ell$ be the rank of
  $\emodule_0$. Let $\{e_1,\dots, e_{\ell}\}$ be the standard basis for
  $\emodule_0$, and $p_i=\epsilon(e_i)$ for $1\le i\le\ell$. Then
  $\ker(\epsilon)=\Syz(p_1,\dots, p_{\ell})$.
\end{proposition}

Following \cref{prop:syzygy-resolution}, if we fix a generating set
$\{q_1,\dots, q_t\}$ of $\Syz(\module)=\ker(\epsilon)$, then we can take
$\emodule_1=\ring^t$ and, as a matrix, \(d_1 = (q_{ij})_{1 \le i \le t, 1 \le j
\le \ell}\).
Continuing in this fashion, we construct \(d_2,\ldots,d_m\) such that
$\Syz_{j+1}(\module)=\ker(d_{j})$ for $1\le j\le m$.

\subsection{The matrix-\texorpdfstring{F\textsubscript{5}}{F5} algorithm}
\label{sec:prelim:matrix-f5}

The matrix-$F_5$ algorithm \cite{BardetFaugereSalvy2015} is based on $F_5$
\cite{Faugere2002}. For the needs of this paper, we describe here
a version of the former which exploits a more general \emph{syzygy criterion}
of the latter, as explained below.

Throughout, we will take $\ord$ to be the grevlex monomial order on $\ring$,
and $\ordpot$ to be the \textit{position over term} order on the free module
$\ring^{t}$, for any $t\ge 1$. That is, for monomials $x =
(0,\ldots,0,x_i,0,\ldots,0)$ and $y = (0,\ldots,0,y_j,0,\ldots,0)$ in
$\ring^{t}$ with respective supports \(i\) and \(j\), $x\ordpot y$ if
and only if \(i < j\) or (\(i = j\) and \(x_i\ord y_j\)).

\subsubsection{Macaulay matrices; signatures}
We take the standard grading by degree on $\ring$, which
  induces a grading on the free module $\ring^{t}$ for any $t\in\ZZp$.
Let $F=(f_1,\dots,f_\ell)\subseteq\ring^{t}$ be a sequence of homogeneous elements of $\ring^{t}$. That is, for each $1\le i\le\ell$, all coordinates of $f_i$ (with respect to the standard basis of $\ring^{t}$) are homogeneous of the same degree. We
assume $d_1\le d_2\le\dots\le d_\ell$, where $d_i=\deg(f_i)$, without loss of
generality. For $d\ge d_1$ and $1\le i\le\ell$, let $\macmat_{d,i}$ be the
Macaulay matrix of $(f_1,\dots,f_i)$ in degree $d$. Each row of $\macmat_{d,i}$
corresponds to a polynomial $\tau f_j$ where $1\le j\le i$, $d_j\le d$, and
$\tau$ is a monomial of degree $d-d_j$; the pair $(j,\tau)$ is called the
\emph{signature} of this row. The columns of $\macmat_{d,i}$ are indexed by the monomials of $\ring^{t}$ of degree $d$, and are ordered in decreasing order with respect to $\ordpot$.
We take a position over term order \(\ordsig\)
on the set of pairs $(j,\tau)$ with $1\le j\le\ell$ and $\tau$ a monomial of
$\ring$:
\[
	(j',\tau') \ordsig (j,\tau)
  \quad\text{if}\quad
  j'<j \text{ or } (j'=j \text{ and } \tau' \ord \tau).
\] 
A \emph{valid row operation} on $\macmat_{d,i}$ consists in adding to a row
with signature $(j,\tau)$ some $\Bbbk$-multiple of a row with signature which
is \(\ordsig\)-less than \((j,\tau)\). We denote by $\redmacmat_{d,i}$ any row
echelon form of $\macmat_{d,i}$ obtained via a sequence of valid row
operations. We will denote by $\lt(\redmacmat_{d,i})$ the monomials
corresponding to the pivot columns of $\redmacmat_{d,i}$. Recall that the
$f_1,\dots, f_\ell$ are homogeneous. The nonzero rows of
$\redmacmat_{d,i}$ therefore form the elements of degree $d$ of a Gröbner basis
for $\genby{f_1,\dots,f_i}$. For an integer $D\ge 0$, a set $G$ is called a
$D$\textit{-Gröbner basis} for $\genby{F}$ if for all elements $f\in\genby{F}$
of degree at most $D$, $\ltpot(f)\in\ltpot(\genby{G})$. Thus, a $D$-Gröbner
basis for $\module=\genby{F}$ is obtained by computing $\redmacmat_{d,\ell}$ for all $d_1\le d\le D$.  Note that when $t=1$,
$f_1,\dots, f_\ell$ are polynomials, and $\module=\genby{F}$ is simply a
homogeneous ideal of $\ring$, whence the rows of $\redmacmat_{d,i}$ form the
elements of degree $d$ of a traditional Gröbner basis for $\genby{f_1,\dots,
f_i}$.

\subsubsection{The syzygy criterion} 

When there are syzygies amongst $\bm{f} = (f_1,\dots,
  f_\ell)$, the Macaulay matrices $\macmat_{d,i}$ do not have full rank. With
prior knowledge of these syzygies, the matrix-$F_5$ algorithm can avoid rows
which reduce to zero when computing $\redmacmat_{d,i}$ from $\macmat_{d,i}$.

\begin{proposition}[Syzygy Criterion, {\cite[Lem.~6.4]{EderFaugere2016}}]\label{prop:syz-crit}
	Let $s=(s_1,\dots, s_\ell)$ be a homogeneous syzygy of $\bm{f}$ and $\ltpot(s)=\tau e_i$. Then 
  \begin{enumerate}[leftmargin=0.5cm]
		\item \label{it:syz-crit-1} The row of $\macmat_{\deg\tau+d_i,
			i}$ with signature $(i,\tau)$ is a linear combination
			of rows of $\macmat_{\deg\tau+d_i,i}$ of smaller
			signature.
		\item \label{it:syz-crit-2}
			For any monomial $\sigma\in\ring$, the row of
			$\macmat_{\deg\tau+\deg\sigma+d_i,i}$ 
			with signature $(i,\sigma\tau)$ is a linear
			combination of rows of
			$\macmat_{\deg\tau+\deg\sigma+d_i,i}$ of smaller
			signature.
	\end{enumerate}
\end{proposition}
\begin{proof}
  We have $\tau f_i=-\sum_{j\ne
  i}s_jf_j-f_i(s_i-\ltpot(s))$.  The module element $\tau f_i$ corresponds to the
  row of $\macmat_{\deg\tau+d_i,i}$ with signature $(i,\tau)$, while $\sum_{j\ne
  i}s_jf_j-f_i(s_i-\ltpot(s))$ is a \(\field\)-linear combination of other rows
  of $\macmat_{\deg\tau+d_i, i}$.  This proves \cref{it:syz-crit-1}.
  
  Suppose now that the row with signature $(i,\tau)$ of
  $\redmacmat_{\deg\tau+d_i,i}$ is a zero row. Then the polynomial $\tau f_i$ is
  a $\Bbbk$-linear combination of rows of $\macmat_{\deg\tau+d_i,i}$ with smaller
  signature, i.e.,
  \[
    \tau f_i
    =
    \sum_{(i',\tau') \ordsig (i,\tau)} c_{(i',\tau')} \tau' f_{i'}
    \text{ for some } c_{(i',\tau')} \in\Bbbk
    .
  \] 
  We can write $\sigma\tau f_i = \sum_{(i',\tau') \ordsig (i,\tau)}
  c_{(i',\tau')} \sigma \tau' f_{i'}$, for any monomial $\sigma$ in \(\ring\).
  Hence, the row with signature $(i,\sigma\tau)$ of
  $\macmat_{\deg\tau+\deg\sigma+d_i,i}$ is a \(\field\)-linear combination of
  rows with smaller signature. 
\end{proof}
  
If $t=1$, the Koszul syzygies $f_jf_i-f_if_j=0$ for all $1\le i,j\le\ell$ always
exist, and produce linear dependencies between the rows of the Macaulay
matrices. The matrix-$F_5$ algorithm works by interpreting these syzygies in
this way to predict the signatures of rows which will reduce to zero when
computing $\redmacmat_{d,i}$ from $\macmat_{d,i}$, and avoiding such rows
altogether. Succinctly, this algorithm utilizes the following criterion, which
is a specialization of \cref{prop:syz-crit}.

\begin{proposition}[$F_5$ Criterion, {\cite[Thm.\,1]{Faugere2002}}]\label{prop:f5-crit}
	The rows with signature $(i,\tau)$ of
	$\macmat_{d,i}$ reduce to zero in $\redmacmat_{d,i}$, for all
	$\tau\in\lt(\redmacmat_{d-d_i,i-1})$.
\end{proposition}

\subsubsection{The matrix-\texorpdfstring{$F_5$}{F5} algorithm}
\label{sec:prelim:alg-mF5-syz}

When $t=1$, combining the syzygy criterion with \cref{prop:f5-crit} leads to
the matrix-$F_5$ algorithm. It works incrementally by degree and index. That
is, for a fixed degree $d$, it first computes the elements of degree $d$ of a
Gröbner basis for $(f_1)$ by reducing the matrix $\macmat_{d,1}$ to 
$\redmacmat_{d,1}$, and then builds the matrix $\macmat_{d,2}$ using
$\redmacmat_{d,1}$. Continuing in this fashion, it eventually builds and
reduces $\macmat_{d,\ell}$, yielding the elements of degree $d$ of a Gröbner
basis for the full system $F$.

In \cref{alg:mF5-syz}, we complement the description of this algorithm
from~\cite{BardetFaugereSalvy2015} by integrating \cref{it:syz-crit-2} of
\cref{prop:syz-crit}. This is important since it allows us to avoid a
significant number of reductions to zero that would occur without it. We allow
for the input of precomputed syzygies of $F$ in order to exploit
\cref{prop:syz-crit} and we allow $t\ge 1$. The termination and correction
of \cref{alg:mF5-syz} is from \cite[Thm.~9]{BardetFaugereSalvy2015}
 when $t=1$, and the same induction argument works when $t>1$.

\begin{algorithm}[ht]
	\caption{Matrix-\(F_5\)\((F,D,S)\)}
	\label{alg:mF5-syz}
  \begin{algorithmic}[1]
	  \Require{%
      A sequence $F=(f_1,\dots, f_\ell)$ of homogeneous elements of degrees $d_1\le\cdots\le d_\ell$ in \(\pring^{t}\);
      a degree bound $D$; a set $S$ of syzygies of $F$.}
    \Ensure{The reduced POT $D$-Gröbner basis for $\genby{F}$.}

    \State \InlineFor{$i\in\{1,\dots,\ell\}$}{$G_i\gets\emptyset$}
    \For{$d$ from $d_1$ to $D$}
      \State $\macmat_{d,0}\gets\emptyset$; $\mathrm{Crit}\gets\ltpot(S)$
      \For{$i$ from $1$ to $\ell$}
        \If{$d<d_i$}
          $\macmat_{d,i}\gets \macmat_{d,i-1}$
        \ElsIf{$d=d_i$}
          \State 	\hspace{-1.5cm}    $\macmat_{d,i}\gets$ concatenate the row $f_i$ to $\redmacmat_{d,i-1}$ with signature $(i,1)$
        \Else
          \State $\macmat_{d,i}\gets\redmacmat_{d,i-1}$
	  \If{$t=1$}
          \For{$\tau\in\lt(\macmat_{d-d_i,i-1})$}
          \State $\mathrm{Crit}\gets\mathrm{Crit}\cup\{(i,\tau)\}$
          \EndFor
	  \EndIf
	  \For{$f\in\rows(\redmacmat_{d-1,i})\smallsetminus\rows(\redmacmat_{d-1,i-1})$}
            \State $(i,\tau)\gets$ signature of $f$
	    \If{$f=0$}
	    	\For{$j\in\{1,\dots, k\}$}
		\State $\mathrm{Crit}\gets\mathrm{Crit}\cup\{(i,\tau\cdot x_j)\}$
		\EndFor
	    \EndIf
	  \EndFor
          \For{$f\in\rows(\macmat_{d-1,i})\smallsetminus\rows(\macmat_{d-1,i-1})$}
            \State $(i,\tau)\gets$ signature of $f$
	    \For{$j\in\{\max\{j':x_{j'}\mid\tau\},\dots,k\}$}\label{line:var-select}
              \If{$(i,\tau\cdot x_j)\notin\mathrm{Crit}$}
                \State \hspace{-3cm} $\macmat_{d,i}\gets$ concatenate the row $x_jf$ to $\macmat_{d,i}$ with signature $(i,\tau\cdot x_j)$\label{line:mac-mat-add-row}
              \EndIf
            \EndFor
          \EndFor
        \EndIf
      \State $\redmacmat_{d,i} \gets$ reduced row echelon form of $\macmat_{d,i}$ obtained via a sequence of valid elementary row operations
      \State $G_i\gets G_i\cup\{f\in\rows(\redmacmat_{d,i}):f\notin \genby{\lt(G_i)}\}$ 
    \EndFor
    \EndFor
    \State \Return $G_\ell$
  \end{algorithmic}
\end{algorithm}

\subsection{Genericity}

We take notation from \cite[Sec.~2 and~3]{FaugereSafeySpaenlehauer2013}. Fix
$n,k\in\ZZ_{>0}$. Define $\fraka=\{\fraka_t^{(i,j)}:1\le t\le k, 1\le i,j\le
n\}$. For each $1\le i,j\le n$, let
$f_{i,j}=\sum_{t=1}^{k}a_t^{(i,j)}x_t\in\field[\fraka, x_1,\dots, x_k]$. We call
$f_{i,j}$ a \textit{generic homogeneous linear form}. We denote by $\genmat$ the
matrix over $\field[\fraka,x_1,\dots,x_k]$ whose $(i,j)$ entry is $f_{i,j}$.
Next, for a fixed $\coeffa=\left(a_t^{(i,j)}\right)\in\closure^{k\cdot n^2}$, we denote by
$\varphi_{\coeffa}$ the specialization map
$\varphi_{\coeffa}:\field[\fraka,x_1,\dots, x_k]\to\closure[x_1,\dots, x_k]$
which specializes $\fraka_t^{(i,j)}$ to $a_t^{(i,j)}$. We call a map
\[
	\property:\operatorname{Ideals}(\field[\fraka, x_1,\dots, x_k])\to\{\mathtrue, \mathfalse\}.
\] 
a \textit{property}.
For an integer $1\le r<n$, we will denote by $\gendetideal$ the ideal of
$(r+1)$-minors of $\genmat$. Subsequently, a property $\property$ is called
$\gendetideal$\textit{-generic} if there exists a nonempty Zariski open subset $U$ of
$\affspace$ such that for all $\coeffa\in U$, $\property\left(
  \varphi_{\coeffa}(\gendetideal)  \right) = \mathtrue$.

An important property is the notion of Cohen-Macaulayness. Let $\ideal$ be an
ideal of $\ring$. A sequence $(f_1,\dots, f_\ell)\subseteq\ring$ is called an
$\ideal$-\textit{regular sequence} if for all $1\le i\le\ell$, $f_i$ is not a
zero-divisor in the module $\ideal/\genby{f_1,\dots, f_{i-1}}$. The ideal
$\ideal$ is called \textit{Cohen-Macaulay} if there exists an $\ideal$-regular
sequence $(f_1,\dots, f_\ell)$ such that $\ell=\dim(\ideal)$ (here
$\dim(\ideal)$ is the Krull dimension of $\ideal$ in $\ring$).

\begin{remark}
	If $(f_1,\dots, f_\ell)$ is an $\ideal$-regular sequence, then
  $\ell\le\dim(\ideal)$. Hence, Cohen-Macaulayness requires that there exists an
  $\ideal$-regular sequence of maximal possible length in $\ring$.
\end{remark}

\begin{proposition}\label{prop:cm-generic}
  Let $\CM$ be the property
  \(\CM(\ideal)=\mathtrue\) if \(\ideal\) is Cohen-Macaulay and
  \(\CM(\ideal)=\mathfalse\) otherwise. Then for any $1\le r\le n-2$, $\CM$ is
  $\gendetideal$-generic.
\end{proposition}
\begin{proof}
	Let $U$ be an $n\times n$ matrix of indeterminates;
  $\detidealGen{U}{r}$ is Cohen-Macaulay \cite[Thm.~2.5]{BrunsVetter1988};
  \cite[Lem.~3]{FaugereSafeySpaenlehauer2013} ends the proof.
\end{proof}

\section{Syzygies of determinantal ideals}
\label{sec:syz}

Here, we focus on the syzygies between the minors \(\detsystem\) of order $r+1$
of $M$. The module $\Syz(\detsystem)$
is known to be generated by syzygies between minors of order $r+1$ of submatrices of $M$
of size $(r+2)\times (r+2)$ \cite[Thm.\,5.1]{kurano1989}. This allows us to reduce the problem of computing
generators for $\Syz(\detsystem)$ from the general case to the case $r=n-2$.
The Gulliksen-Neg{\aa}rd complex \cite{GulliksenNegard1972,BrunsVetter1988} is
a free resolution of $\detidealCorkOne$. We will exploit this complex to obtain
$\Syz(\detsystem)$ first when $r=n-2$, then in full generality. 

\subsection{The Gulliksen-Neg{\aa}rd complex}\label{subsec:gulliksen-negard}

The Gulliksen-Neg{\aa}rd complex is a free resolution of $\detidealCorkOne$,
\[
	0\to \emodule_3\xrightarrow{d_3} \emodule_2\xrightarrow{d_2} \emodule_1\xrightarrow{d_1} \emodule_0\xrightarrow{\epsilon} \detidealCorkOne\to 0
  .
\]
As such, we can use \cref{prop:syzygy-resolution} to compute the first syzygy
module of the set of generators \(\detsystemCorkOne\) as the kernel of the
augmentation map $\epsilon$ of this complex. We recall
the construction of the complex here; details and proofs can be found in
\cite[2.D]{BrunsVetter1988}.

We denote by $\matmod$ the set of $n\times n$ matrices over $\ring$, with the
structure of a free $\ring$-module of rank $n^2$. We will denote by $\matringbasis{i}{j}$
the standard $(i,j)$-th basis matrix of $\matmod$. In this section we will take
as generators for $\detidealCorkOne$ the cofactors of $M$. To that
end, let $\cofac{M} = (\cofac{M}_{i,j})_{i,j} \in \matmod$ be the matrix of
these cofactors.

\subsubsection{The modules}

We begin by defining the component modules $\emodule_3$, $\emodule_2$,
$\emodule_1$, $\emodule_0$. Let $\emodule_0=\matmod$. Consider the sequence
\[
	\ring
  \xrightarrow{\iota} \matmod \oplus \matmod
  \xrightarrow{\pi} \ring
\] 
with $\iota(a)=(aI_n, aI_n)$, where $I_n$ is the identity matrix in $\matmod$
and $\pi(X,Y) = \trace(X-Y)$ is the trace of $X-Y$. The module $\ker(\pi)$ is
generated by the union of the following sets:
\begin{itemize}
  \item \(\{(0,\matringbasis{i}{j})\in \matmod\oplus\matmod : 1\le i,j\le n, i\ne j\}\),
  \item \(\{(\matringbasis{i}{j},0)\in \matmod\oplus\matmod : 1\le i,j\le n, i\ne j\}\),
  \item \(\{(\matringbasis{i}{i},\matringbasis{1}{1})\in \matmod\oplus \matmod : 1\le i\le n\}\), and
  \item \(\{(0, \matringbasis{i}{i}-\matringbasis{1}{1})\in \matmod\oplus \matmod : 2\le i\le n\}\).
\end{itemize}

On the other hand, $\im(\iota)$ is generated by
\[
  (I_n,I_n) 
  = (\matringbasis{1}{1},\matringbasis{1}{1})
  + \textstyle\sum_{i=2}^{n} (\matringbasis{i}{i},\matringbasis{1}{1})
  + \textstyle\sum_{i=2}^{n} (0,\matringbasis{i}{i}-\matringbasis{1}{1})
  .
\]
This shows that $\emodule_1=\ker(\pi)/\im(\iota)$ is a free module. Finally, let
$\emodule_2=\matmod$ and $\emodule_3=\ring$.

\subsubsection{The maps}

We next define the maps $d_1$, $d_2$, $d_3$, $\epsilon$, as follows:
\begin{itemize}
  \item \(\epsilon: \emodule_0 \to \detidealCorkOne, \;\;	N \mapsto \trace(\cofac{M}N)\),
  \item \(d_1: \emodule_1 \to \emodule_0, \;\; \overline{(N_1,N_2)} \mapsto N_1M-MN_2\),
  \item \(d_2: \emodule_2 \to \emodule_1, \;\; N \mapsto \overline{(MN, NM)}\), and
  \item \(d_3: \emodule_3 \to \emodule_2, \;\; x \mapsto x\cofac{M}\),
\end{itemize}
where for $(N_1,N_2)\in \matmod\oplus \matmod$, we denote by
$\overline{(N_1,N_2)}$ its image under the canonical surjection $\matmod\oplus
\matmod\twoheadrightarrow \emodule_1$. 

\begin{proposition}\label{prop:gn-free-resolution}
	Let $M$ be a matrix of homogeneous linear forms in $\ring$.  Assume
	$\detidealCorkOne$ is Cohen-Macaulay. With
	\[
	\emodule_0,\emodule_1,\emodule_2,\emodule_3,\epsilon,d_1,d_2,d_3
	.\] 
	as defined above, the sequence
	\[
		0\to \emodule_3\xrightarrow{d_3} \emodule_2\xrightarrow{d_2} \emodule_1\xrightarrow{d_1} \emodule_0\xrightarrow{\epsilon} \detidealCorkOne\to 0
	\] 
	is a free resolution of $\detidealCorkOne$.
\end{proposition}
\begin{proof}
	Since $\mathcal{I} = \detidealCorkOne$ is Cohen-Macaulay, there exists an
  $\mathcal{I}$-regular sequence of length equal to the Krull dimension of
  $\mathcal{I}$ in $\ring$. By \cite[Thm.~10]{FaugereSafeySpaenlehauer2013} and
  \cref{prop:cm-generic}, the Krull dimension of $\mathcal{I}$ is exactly $4$.
  Then, the result follows from \cite[Thm.~2.26]{BrunsVetter1988}.
\end{proof}

\subsection{The case \texorpdfstring{$r=n-2$}{r = n-2}}

We give generators for the first syzygy module in the case $r=n-2$, assuming $\detidealCorkOne$ is Cohen-Macaulay.

\begin{theorem}\label{thm:syz-corank-1} 
	Let $M=(\matentry{i}{j})$ be a matrix of homogeneous linear forms in
	$\ring$. Suppose that $\detidealCorkOne$ is Cohen-Macaulay. Then the
	first syzygy module of $F_{n-2}(M)$ is generated by:
  \begin{enumerate}[(i)]
		\item \label{it:thm:syz-corank-1:type1}
      $\sum_{k=1}^{n}(-1)^{k+j}\matentry{k}{i}\matringbasis{k}{j}$
      for $i\ne j$;
		\item \label{it:thm:syz-corank-1:type2}
      $\sum_{k=1}^{n}(-1)^{i+k}\matentry{j}{k}\matringbasis{i}{k}$
      for $i\ne j$;
    \item \label{it:thm:syz-corank-1:type3} 
      $\sum_{k=1}^{n}((-1)^{i+k}\matentry{k}{i}\matringbasis{k}{i}-(-1)^{k+1}\matentry{1}{k}\matringbasis{1}{k})$
      for $1\le i\le n-1$;
		\item \label{it:thm:syz-corank-1:type4}
      $\sum_{k=1}^{n}((-1)^{j+k}\matentry{j}{k}\matringbasis{j}{k}-(-1)^{k+1}\matentry{1}{k}\matringbasis{1}{k})$
      for $2\le j\le n$.
	\end{enumerate}
	Furthermore, the syzygies described by Items {\em (i), (ii), (iii)} and {\em (iv)}
	form a minimal generating set for the $\Syz(F_{n-2}(M))$ of size
	$2n^2-2$.
\end{theorem}

\begin{proof}
  By \cref{prop:syzygy-resolution}, $\ker(\epsilon)$ is the first syzygy module
  of the cofactors of $M$. By \cref{prop:gn-free-resolution}, since
  $\detidealCorkOne$ is Cohen-Macaulay, the Gulliksen-Neg{\aa}rd complex is
  exact and $\ker(\epsilon)=\im(d_1)$. The image $\im(d_1)$ is generated by the
  images of generators for $\emodule_1$ under $d_{1}$. Thus, by
  \cref{subsec:gulliksen-negard}, the first syzygy module of $F_{n-2}(M)$ is
  generated by the following syzygies. For $i\ne j$,
\begin{equation}\label{eq:gn-syz-1}
	d_1\!\left(\overline{(\matringbasis{i}{j}, 0)}\right)=\matringbasis{i}{j}M=\textstyle\sum\limits_{k=1}^{n}\matentry{k}{i}\matringbasis{k}{j}
.\end{equation} 
Similarly, for $i\ne j$,
\begin{equation}\label{eq:gn-syz-2}
	d_1\!\left(\overline{(0,\matringbasis{i}{j})}\right)=M\matringbasis{i}{j}=\textstyle\sum\limits_{k=1}^{n}\matentry{j}{k}\matringbasis{i}{k}
.\end{equation} 
For any $1\le i\le n-1$,
\begin{equation}\label{eq:gn-syz-3}
	d_1\!\left(\overline{(\matringbasis{i}{i},\matringbasis{1}{1})}\right)=E_{i,i}M-M\matringbasis{1}{1}=\textstyle\sum\limits_{k=1}^{n}\matentry{k}{i}\matringbasis{k}{i}-\matentry{1}{k}\matringbasis{1}{k}
.\end{equation} 
Finally, for any $2\le j\le n$, 
\begin{equation}\label{eq:gn-syz-4}
	d_1\!\left(\overline{(0,\matringbasis{j}{j}-\matringbasis{1}{1})}\right)=M\matringbasis{j}{j}-M\matringbasis{1}{1}=\textstyle\sum\limits_{k=1}^{n}\matentry{j}{k}\matringbasis{j}{k}-\matentry{1}{k}\matringbasis{1}{k}
  .\end{equation}
Since the generators for $\detidealCorkOne$ taken in the
Gulliksen-Neg{\aa}rd complex are the cofactors of $M$ rather than the
$(n-1)$-minors of $M$, we obtain \crefrange{it:thm:syz-corank-1:type1}{it:thm:syz-corank-1:type4}
by pulling back each of \crefrange{eq:gn-syz-1}{eq:gn-syz-4},
 respectively under
the isomorphism $\cofac{M}_{i,j} \in \detidealCorkOne\mapsto (-1)^{(i+j)}\cofac{M}_{i,j}\in \detidealCorkOne$.
There are $n^2-n$ syzygies described by each of
\cref{it:thm:syz-corank-1:type1} and \cref{it:thm:syz-corank-1:type2}, and $n-1$
syzygies described by each of \cref{it:thm:syz-corank-1:type3} and
\cref{it:thm:syz-corank-1:type4}. This gives a total of $2n^2-2$ syzygies.

We conclude by proving that these $2n^2-2$ syzygies form a minimal generating
set for $\Syz(F_{n-2}(M))$. Let $m_1,\dots, m_{2n^2-2}\in \Syz(F_{n-2}(M))$
denote the generating set given by \cref{it:thm:syz-corank-1:type1},
\cref{it:thm:syz-corank-1:type2}, \cref{it:thm:syz-corank-1:type3},
\cref{it:thm:syz-corank-1:type4}. Suppose that for some $1\le i\le 2n^2-2$,
$m_i$ is generated by $\{m_1,\dots, m_{2n^2-2}\}\smallsetminus\{m_i\}$. Then we
can write  $\sum_{j\ne i}a_jm_j=m_i$ for some $a_j\in\ring$. Since the $m_j$ are all homogeneous, this
forces $a_j\in\Bbbk$ for all $j\ne i$. Subsequently,
$a_j\in\im(d_2)\cap\Bbbk^{2n^2-2}$ for all $j\ne i$. \\
Letting $a_i=1$, and taking $N\in d_2^{-1}((a_1,\dots, a_{2n^2-2}))$, we find
that $MN,NM\in\matmod$ are matrices with entries in $\Bbbk$. For $1\le j\le n$,
the entries of the $j$-th row of $MN$ are members of the ideal generated by the
$j$-th row of $M$. The entries of $M$ are homogeneous linear forms, so the only
constant element contained in this ideal $0$. Similarly, for $1\le j\le
n$, the entries of the $j$-th row of $NM$ are members of the ideal generated by
the $j$-th column of $M$, and an analogous argument applies.  Thus, $a_j=0$ for
each $1\le j\le 2n^2-2$.
\end{proof}

One can easily construct an algorithm, named $\operatorname{SyzCorankOne}$,
which, given a matrix $M$, computes the syzygies described in
\cref{thm:syz-corank-1}.

\begin{remark}
	In both \cref{thm:syz-corank-1} and
	\ifnum\ALGS=1{\cref{alg:syzgen-corank-one}}\else{$\operatorname{SyzCorankOne}(M)$}\fi\
	we require that $\detsystemCorkOne$ is Cohen-Macaulay. This is
	necessary, as without it the Gulliksen-Neg{\aa}rd complex need not be
	exact and subsequently we cannot compute $\Syz(\detsystemCorkOne)$
	using its differential maps. However, since $\epsilon$ is defined by
	$\epsilon(N)=\trace(\cofac{M}N)$, where $\cofac{M}=(\cofac{M}_{i,j})$
	is the matrix of cofactors of $M$, a matrix $N=(N_{i,j})\in\matmod$ is
	in the kernel of $\epsilon$ if and only if $\sum_{1\le i,j\le
	n}N_{j,i}\cofac{M}_{i,j}=0$. That is, $\ker(\epsilon)$ corresponds to
	$\Syz(\detsystemCorkOne)$ even if $\detidealCorkOne$ is not
	Cohen-Macaulay. Moreover, even if $\detidealCorkOne$ is not
	Cohen-Macaulay, the Gulliksen-Neg{\aa}rd complex is still a complex.
	Thus, in all cases, $\im(d_1)\subseteq\ker\epsilon$, so if
	$\detidealCorkOne$ is not Cohen-Macaulay, \cref{thm:syz-corank-1}
	describes (and subsequently
	\ifnum\ALGS=1{\cref{alg:syzgen-corank-one}}\else{$\operatorname{SyzCorankOne}(M)$}\fi\
	computes) a generating set for a submodule of
	$\Syz(\detsystemCorkOne)$.
\end{remark}

\begin{remark}
	If the entries of $M$ are not homogeneous, then assuming 
  $\detidealCorkOne$ is Cohen-Macaulay, the syzygies computed by
  \cref{thm:syz-corank-1} still generate $\Syz(\detsystemCorkOne)$, but they
  need no longer be a \textit{minimal} generating set.
\end{remark}

\subsection{The general case}


\begin{theorem}\label{thm:syz-general}
	Let $n\ge 3$ and let $1\le r\le n-2$. Then there
	exists a nonempty Zariski open set $U\subseteq\affspace$ such that for
	all $\coeffa\in U$, taking $M=\varphi_{\coeffa}(\genmat)$, the
	following holds: Let $M'$ be the set of submatrices of size
	$(r+2)\times (r+2)$ of $M$. For each matrix $N\in M'$, let $S(N)$ be
	the set of syzygies of $\detsystemGen{N}{r}$ computed using
\cref{thm:syz-corank-1}. Then $\Syz(\detsystem)=\bigcup_{N\in M'}S(N)$.
\end{theorem}
\begin{proof}
	Let $\mathscr{U}$ be an $n\times n$ matrix of indeterminates over $\field$.
  Let $\mathscr{U}'$ be the set of $(r+2)\times(r+2)$ submatrices of
  $\mathscr{U}$. For each $\mathscr{N}\in \mathscr{U}'$, let $S(\mathscr{N})$ be
  the set of syzygies of $\detsystemGen{\mathscr{N}}{r}$ computed using
  \cref{thm:syz-corank-1}. By \cite[Thm.\,5.1]{kurano1989},
  $\Syz(\gendetsystem)=\bigcup_{\mathscr{N}\in \mathscr{U}'}S(\mathscr{N})$.
  Thus, by \cite[Lem.\,3]{FaugereSafeySpaenlehauer2013}, there is a nonempty
  Zariski open subset $U_1\subseteq\affspace$ such that for all $\coeffa\in
  U_1$, the syzygies between the $(r+1)$-minors of $\varphi_{\coeffa}(\genmat)$
  are those between the $(r+1)$-minors of each $(r+2)\times (r+2)$ submatrix of
  $\varphi_{\coeffa}(\genmat)$. We denote by $\genmat'$ the set of $(r+2) \times
  (r+2)$ submatrices of $\genmat$. By \cref{prop:cm-generic}, for each submatrix
  $N$ of $\genmat'$, there exists a nonempty Zariski open subset
  $U_N\subseteq\AA^{k\cdot n^2}_{\field}$ such that for all $\coeffa\in U_N$,
  the ideal generated by the $(r+1)$-minors of $N$ is Cohen-Macaulay, so that
  \cref{thm:syz-corank-1} applies. Thus, taking $U=\bigcap_{N\in
    \genmat'}U_N\cap U_1$, the result follows.
\end{proof}

Consequently, using $\operatorname{SyzCorankOne}(M)$
we obtain an algorithm $\operatorname{SyzGen}(M,r)$ which constructs a set of
generators for
$\Syz(\detsystem)$.

\begin{remark}\label{rem:non-generic-general-syzygies}
  From \cref{thm:syz-corank-1,thm:syz-general}, neither
  $\operatorname{SyzGen}(M,r)$ nor $\operatorname{SyzCorankOne}(M)$ require
  any arithmetic $\Bbbk$-operations.
  
	Again in the statement of \cref{thm:syz-general} we require that $\detideal$
  is Cohen-Macaulay. This is necessary in order for $\Syz(\detideal)$ to be
  computed via the syzygies of $(r+1)$-minors of $(r+2)\times (r+2)$
  submatrices. If $\detideal$ is not Cohen-Macaulay, \cref{thm:syz-general}
  gives a (possibly proper) subset of a generating set for $\Syz(\detideal)$.

  Finally, we require that the entries of $M$ be homogeneous linear forms. Once
  again, the theorem holds if the entries are affine, as long as $\detideal$
  satisfies the stated genericity assumption.
  
  Note that no claim is made as to the minimality of the
    generating set computed in \cref{thm:syz-general}. However, one can show
    that when the entries of $M$ are homogeneous, a minimal generating set can
    be extracted from the set computed in \cref{thm:syz-general} by throwing
    away any element which differs by multiplication by $-1$ from another
    element. This need no longer hold if the entries are affine.
\end{remark}

\section{Determinantal matrix-\texorpdfstring{$F_5$}{F5} algorithm}

In this section, we use the syzygies returned by $\operatorname{SyzGen}(M,r)$ to
avoid reductions to zero when running \cref{alg:mF5-syz}. As explained below,
the following result will be instrumental.

\begin{proposition}[{\cite[Lem.\,6.4]{EderFaugere2016}}]\label{prop:gb-signatures}
	Let $(f_1,\dots,f_\ell)=F\subseteq \ring^{t}$
	be a system of homogeneous module elements.
  Let $D\in\ZZ_{\ge 0}$, and let $G=G_{D-\min_i\{\deg(f_i)\}}$
	be the elements up to degree $D-\min_i\{\deg(f_i)\}$ of a POT-Gröbner basis
  for $\Syz(F)$.
  Then,
  \begin{enumerate}[leftmargin=0.5cm]
	  \item\label{it:gb-signatures-1} If $\tau e_i\in\ltpot(G)$, the row of $\macmat_{\deg(\tau)+\deg(f_i), i}$ with signature $(i,\tau)$ is a linear combination of rows with smaller signature.
	  \item\label{it:gb-signatures-2} If a row with signature $(i,\tau)$ of $\macmat_{\deg(\tau)+\deg(f_i),i}$ reduces to zero, then $\tau e_i$ is in the module generated by $\ltpot(G)$.
  \end{enumerate}
\end{proposition}

\begin{proof}
	\cref{it:gb-signatures-1} is simply \cref{prop:syz-crit}. We turn to
  \cref{it:gb-signatures-2}. Fix $\min_i\{\deg(f_i)\}\le d\le D$ and $1\le
  i\le\ell$. Suppose that the row with signature $(i,\tau)$ reduces to zero in
  $\redmacmat_{\deg(\tau)+\deg(f_i),i}$. Then there is a linear dependency $s_1f_1+\dots+s_\ell
  f_\ell=0$.
  This corresponds to a syzygy $s=s_1e_1+\dots+s_\ell e_\ell\in\Syz(F)$
	with $\ltpot(s)=\tau e_i$. Finally,
	\begin{align*}
		\deg(s_i)&=d-\deg(f_i) \le D-\deg(f_i)\le D-\min_i\{\deg(f_i)\}
	.\end{align*}
  for each $1\le i\le\ell$.
	Thus $\ltpot(s)=\tau e_i$ is in $\genby{\ltpot(G)}$.
\end{proof}

Using \cref{prop:gb-signatures}, in order to remove all reductions to zero when
running \cref{alg:mF5-syz} to compute a $D$-Gröbner basis for a graded module
$F\subseteq\ring^{t}$, we compute the leading terms of the elements up to
degree $D-\min_{f\in F}\{\deg f\}$ of a Gröbner basis for $\Syz(F)$. 
We can compute them by running \cref{alg:mF5-syz} on a set of chosen
generators for $\Syz(F)$ itself, with the appropriate degree bound given by
\cref{prop:gb-signatures}. However, if $\Syz_2(F)\ne \{0\}$, then
\cref{prop:gb-signatures} once again shows that reductions to zero will be
encountered when computing the elements up to degree $D-\min_{f\in F}\{\deg f\}$
of a Gröbner basis for $\Syz(F)$.

When $r=n-2$, the Gulliksen-Neg{\aa}rd complex allows us to explicitly compute
generating sets for all higher syzygy modules. Thus, we can avoid all reductions
to zero when computing a $D$-Gröbner basis for $\detsystem$. When $r<n-2$, we
can only compute a generating set for the first syzygy module
$\Syz(\detsystem)$, and thus cannot efficiently remove all reductions to zero.

Now we are ready to describe an algorithm which exploits the syzygies computed by
$\operatorname{SyzGen}(M,r)$ to compute a grevlex Gröbner basis for $\detsystem$ without
reductions to zero in degree $r+2$.

\begin{algorithm}[h]
	\caption{Determinantal-Matrix-$F_5(M,r,D)$}
	\label{alg:det-F5}
  \begin{algorithmic}[1]
	  \Require{An integer $1\le r\le n-2$, an $n\times n$ matrix $M$ of homogeneous linear forms over $\Bbbk$ in $(n-r)^2$ variables such that $\detideal$ is Cohen-Macaulay, and a degree bound $D$.}
	  \Ensure{A grevlex $D$-Gröbner basis for $\detideal$.}
	  \State $S\gets\operatorname{SyzGen}(M,r)$\label{line:syzgen-det-F5}
	  \State $S'\gets\operatorname{Matrix-F_5}(S, 1, \emptyset)$\label{line:mF5-det-F5-1}
	  \State $G\gets\operatorname{Matrix-F_5}(\detsystem, D, S')$\label{line:mF5-det-F5-2}
	  \State \Return $G$
  \end{algorithmic}
\end{algorithm}

\begin{proposition}
	\cref{alg:det-F5} terminates and is correct.
\end{proposition}

\begin{proof}
	Termination follows from that of
  $\operatorname{SyzGen}(M,r)$ and
  \cref{alg:mF5-syz}. To prove correctness, we need to show that the set $S'$
  of \cref{line:mF5-det-F5-1} is indeed a set of syzygies between the
  elements of $\detideal$. By \cref{thm:syz-general}, the set $S$ computed on
  \cref{line:syzgen-det-F5} is a minimal generating set for $\Syz(F_r(M))$. By
  the construction of this generating set, given in \cref{thm:syz-general}, each
  element of $S$ is homogeneous of degree one. Hence, according to
  \cref{sec:prelim:alg-mF5-syz}, the set $S'$ consists of the elements of degree one
  of a POT-Gröbner basis for $\Syz(\detsystem)$.
\end{proof}

\begin{remark}
	Both the number of rows and the number of columns of the Macaulay
	matrix in degree one for the set $S$ on \cref{line:mF5-det-F5-1} of
	\cref{alg:det-F5} is bounded by the number of rows of the Macaulay
	matrix for $\detsystem$ in degree $r+1$. Therefore, asymptotically, the
	arithmetic cost of \cref{alg:det-F5} is bounded by the arithmetic cost
	of its final step, computing the Gröbner basis of $\detsystem$.
\end{remark}

\begin{proposition}
	Let $n\ge 3$, let $1\le r\le n-2$, let $D=r\cdot(n-r)+1$, and let
	$k=(n-r)^2$. There exists a nonempty Zariski open set
	$U\subseteq\affspace$ such that for all $\coeffa\in U$, taking
	$M=\varphi_{\coeffa}(\genmat)$, upon running \cref{alg:det-F5} with
	arguments $M,r,D$:
	\begin{enumerate}
		\item a full grevlex Gröbner basis is returned; and\label{it:alg-det-F5-1}
		\item for each $1\le i\le\binom{n}{r+1}^2$, the matrix $\macmat_{r+2,i}$ has full rank.\label{it:alg-det-F5-2}
	\end{enumerate}
\end{proposition}
\begin{proof}
	By \cite[]{FaugereSafeySpaenlehauer2013}, there exists a Zariski open subset
  $U_1\subseteq\affspace$ such that the maximal degree of a polynomial in the
  reduced grevlex Gröbner basis for $\detideal$ is precisely $D$. Let $U_2$ be a
  nonempty Zariski open subset of $\affspace$ such that the results of
  \cref{thm:syz-general} hold. Let $U=U_1\cap U_2$. \cref{it:alg-det-F5-1}
  follows immediately from the degree bound given in
  \cite{FaugereSafeySpaenlehauer2013}. We turn to \cref{it:alg-det-F5-2}. By
  \cref{prop:gb-signatures}, \cref{it:gb-signatures-2}, it suffices to show that
  the leading terms of the set $S'$ of \cref{line:mF5-det-F5-1} consists of the
  elements of degree at most $r+2$ of $\ltpot(\Syz(\detideal))$. This is
  immediate from \cref{thm:syz-general} and \cref{sec:prelim:alg-mF5-syz}.
\end{proof}

\begin{remark}\label{rem:syz-1-non-generic}
If we do not impose the genericity assumption on $\detideal$ \cref{alg:det-F5}
will still return a $D$-Gröbner basis for $\detideal$, though $\macmat_{r+2,i}$
need no longer be full rank for all $1\le i\le\binom{n}{r+1}^2$.

If the entries of $M$ are affine, by \cref{rem:non-generic-general-syzygies},
there are two possibilities. First,
$\operatorname{SyzGen}(M,r)$  
still returns a generating set for the first syzygy module of $\detsystem$, and
these may be used in the original $F_5$ algorithm which works on affine input
to avoid reductions to zero.
Alternatively, following \cite[Ch.~8, \S~2, Prop.~7]{CoxLittleOShea2007}, one
can simply homogenize $\detsystem$ with respect to a variable $h$ which is
taken to be grevlex smaller than all other variables of $\ring$, and specialize
$h=1$ upon termination of \cref{alg:det-F5}.
\end{remark}

\section{The case \texorpdfstring{$r=n-2$}{r = n-2}}

Now, we describe an altered version of the $F_5$ algorithm which
computes a Gröbner basis for $\detideal$ when $r=n-2$ without any reductions to
zero. Note that \cref{alg:det-F5} does not require $r<n-2$. Thus, we could
simply compute a Gröbner basis for $\detideal$ using \cref{alg:det-F5} when
$r=n-2$. However, only those reductions to zero arising from syzygies of degree
$r+2$ will be avoided. By \cref{prop:gb-signatures}, any syzygies of degree
$d>r+2$ which cannot be generated by the syzygies of degree $r+2$ will manifest
as reductions to zero in the Macaulay matrices in degree $d$.

\subsection{Higher syzygy modules}

By \cref{prop:gn-free-resolution}, the Gulliksen-Neg{\aa}rd complex is a free
resolution of $\detideal$ as soon as $\detideal$ is Cohen-Macaulay. Thus, the
kernels of its differential maps are precisely the syzygy modules of
$\detideal$. The map $d_3$ is defined by $d_3(x)=x\cofac{M}$, where $\cofac{M}$
is the matrix of cofactors of $M$. The third syzygy module $\Syz_3(\detideal)$
is the image of $d_3$, and is thus free of rank $n^2$ and principally generated
by the entries of $\cofac{M}$.

\begin{proposition}\label{prop:syz-2-corank-1}
	Let $M$ be an $n\times n$ matrix of homogenoeus linear forms in
	$\ring$. Suppose $\detidealCorkOne$ is Cohen-Macaulay. In the
	$\ring$-basis for $\ker(\pi)/\im(\iota)$ given in
	\cref{subsec:gulliksen-negard}, the second syzygy module
	$\Syz_2(F_r(M))$ is generated by the following syzygies:
  \begin{enumerate}[(i)]
		\item \label{it:syz-2-corank-1:type1}
      For $2\le i\le n$ and  $1\le j\le n-1$,
			\begin{align*}
				\textstyle\sum\limits_{k\ne j} & m_{k,i}\overline{(E_{k,j}, 0)}+\textstyle\sum\limits_{k\ne i}m_{j,k}\overline{(0, E_{i,k})}  \\
					     & + m_{j,i}\left(\overline{(E_{j,j}, E_{1,1})}+\overline{(0, E_{i,i}-E_{1,1})}\right)
			.\end{align*} 
		\item \label{it:syz-2-corank-1:type2}
      For $2\le i\le n$, 
			\begin{align*}
				\textstyle\sum\limits_{k\ne n} & m_{k,i}\overline{(E_{k,n}, 0)}+\textstyle\sum\limits_{k\ne i}m_{n,k}\overline{(0, E_{i,k})}  \\
					      & - m_{n,i}\left(\textstyle\sum\limits_{j=1}^{n-1}\overline{(E_{j,j}, E_{1,1})}+\textstyle\sum\limits_{j=2}^{n-1}\overline{(0, E_{j,j}-E_{1,1})}\right)
			\end{align*}
		\item  \label{it:syz-2-corank-1:type3}
      For $1\le j\le n-1$,
			\begin{align*}
				\textstyle\sum\limits_{k\ne j} & m_{k,1}\overline{(E_{k,j}, 0)}+\textstyle\sum\limits_{k\ne 1}m_{j,k}\overline{(0, E_{1,k})} 
                    +m_{j,1}\overline{(E_{j,j}, E_{1,1})}
			\end{align*}
		\item \label{it:syz-2-corank-1:type4}
      Finally,
			\begin{align*}
				\textstyle\sum\limits_{k\ne n} & m_{k,1}\overline{(E_{k,n}, 0)}+\textstyle\sum\limits_{k\ne 1}m_{n,k}\overline{(0, E_{1,k})} \\
					      & - m_{n,1}\left(\textstyle\sum\limits_{j=1}^{n-1}\overline{(E_{j,j}, E_{1,1})}+\textstyle\sum\limits_{j=2}^{n}\overline{(0, E_{j,j}-E_{1,1})}\right)
			\end{align*}
	\end{enumerate}
\end{proposition}

\begin{proof}
The second syzygy module $\Syz_2(\detideal)$ is the image of $d_2$, by
\cref{prop:gn-free-resolution}. The map $d_2$ is defined by
\[
	d_2(N)=\overline{(MN,NM)}
.\] 
Taking $\matringbasis{i}{j}$, $1\le i,j\le n$ to be the canonical $\ring$-basis
for $\matmod$, a basis for $\im(d_2)$ is given by
$\{\overline{(M\matringbasis{i}{j}, \matringbasis{i}{j}M)}\mid 1\le i,j\le n\}$. 
We can express $M\matringbasis{i}{j}$ and $\matringbasis{i}{j}M$ in the canonical $\ring$-basis for $\mathcal{M}_n(\ring)$,

\[
  M\matringbasis{i}{j}=\matentry{j}{i}\matringbasis{j}{j}+\textstyle\sum\limits_{k\ne j}\matentry{k}{i}\matringbasis{k}{j}; \,\,
	\matringbasis{i}{j}M=\matentry{j}{i}\matringbasis{i}{i}+\textstyle\sum\limits_{k\ne i}\matentry{j}{k}\matringbasis{i}{k}.
\]
From this, we express generators for $\Syz_2(\detideal)$ in the
$\ring$-basis for $\ker(\pi)/\im(\iota)$. Doing so gives precisely
\cref{it:syz-2-corank-1:type1,it:syz-2-corank-1:type2,it:syz-2-corank-1:type3,it:syz-2-corank-1:type4}.
\end{proof}

Using \cref{prop:syz-2-corank-1}, one can easily construct an algorithm, which
we will call $\operatorname{Syz2GenCorankOne}(M)$ which constructs the set
$\Syz_2(F_{n-2}(M))$. We use this algorithm in the next section to design a
dedicated $F_5$-type algorithm which performs no reduction to zero when
computing a Gröbner basis of $\mathcal{I}_{n-2}(M)$ when $\detidealCorkOne$ is
Cohen-Macaulay and $k=4$.

\begin{remark}
	Analogous to \cref{rem:syz-1-non-generic}, if $\detidealCorkOne$ is not
	Cohen-Macaulay, the Gulliksen-Neg{\aa}rd complex need not be a free
	resolution of $\detidealCorkOne$, though it is still a complex. Thus,
	even if $\detidealCorkOne$ is not Cohen-Macaulay,
	$\im(d_2)\subseteq\ker(d_1)$, so the syzygies described by
	\cref{prop:syz-2-corank-1} are a subset of a generating set for the
	syzygies between the generators for $\ker\epsilon$ given by
	\cref{thm:syz-corank-1}.
\end{remark}

\subsection{A refined \texorpdfstring{F\textsubscript{5}}{F5} algorithm}

We combine \cref{prop:syz-2-corank-1}, \cref{thm:syz-corank-1}, and
\cref{prop:gb-signatures} to give an algorithm which computes a grevlex Gröbner
basis for $\detsystemCorkOne$ without any reductions to zero, provided
$\detidealCorkOne$ is Cohen-Macaulay. In order to obtain the leading terms of
the first syzygy module, of $\detsystemCorkOne$, we must know which rows will
reduce to zero when echelonizing the Macaulay matrices associated to the first
syzygy module in various degrees. By \cref{prop:gb-signatures}, the signatures
of these rows are precisely the leading terms of a Gröbner basis for the second
syzygy module in the appropriate degree.

Subsequently, applying \cref{prop:gb-signatures} once again, the leading terms
of the first syzygy module in various degrees are precisely the signatures of
the rows which reduce to zero when echelonizing the Macaulay matrices associated
to $\detsystemCorkOne$.

\begin{algorithm}[ht]
	\caption{Determinantal-Corank-One-Matrix-$F_5$}
	\label{alg:det-F5-corank-one}
	\begin{algorithmic}[1]
		\Require{An integer $n\ge 3$, an $n\times n$ matrix of generic homogeneous linear forms over $\Bbbk$ in \(4\) variables, and an integer $D\ge n-1$}
		\Ensure{The elements up to degree $D$ of a grevlex Gröbner basis for $\detidealCorkOne$.}
		\State $S_1\gets\operatorname{SyzCorankOne}(M)$
		\State $S_2\gets\operatorname{Syz2GenCorankOne}(M)$
		\State $S_2'\gets\operatorname{Matrix-F_5}(S_2, D-n, \emptyset)$
		\State $S_1'\gets\operatorname{Matrix-F_5}(S_1, D-n+1, S_2')$\label{line:det-F5-syz-2-gb}
		\State $G\gets\operatorname{Matrix-F_5}(\detsystemCorkOne, D, S_1')$
		\State \Return $G$ 
	\end{algorithmic}
\end{algorithm}

\begin{proposition}
	\cref{alg:det-F5-corank-one} terminates and is correct.
\end{proposition}
\begin{proof}
	Termination is an easy consequence from the one of 
	$\operatorname{SyzCorankOne}(M)$,
	$\operatorname{Syz2GenCorankOne}(M)$, and \cref{alg:mF5-syz}. For
	correctness, it suffices to show that the set $S_1'$ computed on
	\cref{line:det-F5-syz-2-gb} is indeed a set of syzygies of the
	polynomials in $\detsystemCorkOne$. This follows from
	\cref{thm:syz-corank-1}.
\end{proof}

\begin{proposition}\label{prop:alg-det-F5-corank-one-full-rank}
	Let $D=2n-3$. Then there is a nonempty Zariski open subset $U$ of
	$\affspaceCorkOne$ such that for all $\coeffa\in U$, upon running
	\cref{alg:det-F5-corank-one} with arguments
	$\varphi_{\coeffa}(\gendetidealCorkOne),D$,  
	\begin{enumerate}
		\item a full grevlex Gröbner basis is returned; and\label{it:det-F5-corank-one-1}
		\item for each $1\le i\le n^2$ and for each $n-1\le d\le 2n-3$, the matrix $\macmat_{d,i}$ is full rank.\label{it:det-F5-corank-one-2}
	\end{enumerate}
\end{proposition}

\begin{proof}
	By \cite[Lem.~18]{FaugereSafeySpaenlehauer2013}, there is a Zariski dense subset
	$U_1$ of $\affspaceCorkOne$ such that for all $\coeffa\in
	U_1$, the maximal degree of a polynoimal in the reduced grevlex Gröbner
	basis of $\detidealCorkOne$ is $2n-3$. By \cref{prop:cm-generic}, there
	is a nonempty Zariski open subset $U_2$ of $\affspaceCorkOne$
	such that for all $\coeffa\in U_2$, the ideal
	$\varphi_{\coeffa}(\gendetidealCorkOne)$ is Cohen Macaulay. Thus, taking
	$U=U_1\cap U_2$, we obtain \cref{it:det-F5-corank-one-1}.

	We turn to \cref{it:alg-det-F5-2}. By \cref{prop:gb-signatures},
	\cref{it:gb-signatures-2}, it suffices to show that the leading terms
	of the set $S_1'$ computed on \cref{line:det-F5-syz-2-gb} consists of
	the elements of degree at most $2n-3$ of $\ltpot(\Syz(\detidealCorkOne))$.
	This is immediate from \cref{thm:syz-corank-1} and
	\cref{sec:prelim:alg-mF5-syz}.
\end{proof}

\section{Complexity in the case \texorpdfstring{$r=n-2$}{r = n-2}}

Throughout this section we focus on the dimension zero
  case. Thus, $k=(n-r)^2=4$. For a homogeneous ideal $\ideal\subseteq\ring$, we
take $\HF_\ideal(d)$ to be the \textit{Hilbert function} of $\ideal$. That is,
for an integer $d\ge 0$, $\HF_\ideal(d)=\dim_{\field}\ideal_d$. Further, we take
$H_\ideal(t)=\sum_d\HF_\ideal(d)t^{d}$ to be the \textit{Hilbert series} of
$\ideal$. We refer to \cite[1.9]{Eisenbud1995} for further details.

When $r=n-2$, we can use the results of the previous section to give explicit
formulae for the coefficients of the Hilbert series $H_{\detideal}(t)$.
Subsequently, we can exactly compute the ranks of the Macaulay matrices in each
degree computed by the $F_5$ algorithm, and bound the complexity of computing
the reduced grevlex Gröbner basis of a matrix of generic homogeneous linear
forms by that of computing the row reduction of each of these
matrices.

First, note that for any $1\le d\le r$, both the number of rows and the number
of columns of the Macaulay matrix in degree $d$ for the set $S_2$ (resp.\
$S_1$) computed by \cref{alg:det-F5-corank-one} is bounded by the number of
rows of the Macaulay matrix in degree $d+1$ (resp. $d+r+1$) for the set $S_1$
(resp. $\detsystemCorkOne$). Thus, the arithmetic cost of
\cref{alg:det-F5-corank-one} is bounded by the arithmetic cost of the final
step, computing the grevlex Gröbner basis for $\detsystemCorkOne$.

\begin{proposition}\label{prop:det-hilb-series}
	There exists a Zariski open set $U \subseteq \affspaceCorkOne$ such that for all $\coeffa\in U$, the Hilbert series
	$H_{\varphi_{\coeffa}(\gendetideal)}(t)$ for $\varphi_{\coeffa}(\gendetideal)$ is given by:
  \begin{equation}
    \label{eq:hilb-series-det-ideal}
    \sum_{d=r+1}^{2r+1}\left(n^2 \binom{d-r+2}{3}-(2n^2-2) \binom{d-r+1}{3}+n^2\binom{d-r}{3}\right)t^d.
  \end{equation}
\end{proposition}

\begin{proof}
	Let $U$ be as in \cref{prop:alg-det-F5-corank-one-full-rank}. If
	$\module$ is a free $\ring$-module of rank $t$, then the monomials of
	$\module$ of degree $d$ form a basis for the finite-dimensional
	$\Bbbk$-vector space of homogeneous elements of degree $d$ of
	$\module$. Thus, $\HF_{\module}(d)=t\cdot\binom{k+d-1}{d-1}$. The
	description of each free module in the Gulliksen-Neg{\aa}rd complex
	given in \cref{subsec:gulliksen-negard} gives rise to
	\begin{gather*}
		\rk\emodule_0=\#\detsystemCorkOne={\textstyle\binom{n}{n-1}^2}=n^2, \quad \rk\emodule_1=2n^2-2\\
		\rk\emodule_2=n^2, \quad 		\rk\emodule_3=1
	\end{gather*}
	Thus, by \cite[Thm.~1.13]{Eisenbud1995},
  $\HF_{\detidealCorkOne}(d)=\sum_{i=0}^{3}(-1)^{i}\HF_{\emodule_i}(d)$, which
  equals $n^2 \binom{d-r+2}{3}-(2n^2-2) \binom{d-r+1}{3}+n^2\binom{d-r}{3}-\binom{d-r-1}{3}$.
\end{proof}

In the following proposition, we take
\[
\mathscr{B} = \sum_{d=r+1}^{2r+1} \ n^2
\binom{d-r+2}{3}-(2n^2-2)\binom{d-r+1}{3} \ + n^2\binom{d-r}{3}
.\] 
\begin{proposition}\label{prop:corank-one-complexity}
	There is a Zariski dense subset $U$ of $\affspaceCorkOne$ such that
  for all $\coeffa\in U$, the arithmetic cost of computing the reduced grevlex
  Gröbner basis for $\varphi_{\coeffa}(\gendetidealCorkOne)$ using
  \cref{alg:det-F5-corank-one} is in
  \[
    \textstyle
	  O\left(\mathscr{B}^{\omega-1}\binom{2r+5}{5}\right)=O\left(n^{4(\omega-1)}\binom{2n}{3}\right).
  \]
\end{proposition}

\begin{proof}
Take $U$ as in \cref{prop:det-hilb-series}. Fix $\coeffa\in U$ and let
$M=\varphi_{\coeffa}(\gendetidealCorkOne)$. The ideal $\detidealCorkOne$ is
homogeneous, so the complexity of computing a grevlex Gröbner basis for
$\detidealCorkOne$ is bounded by the complexity of reducing the intermediate
Macaulay matrices encountered in the matrix-$F_5$ algorithm. 
The coefficient on $t^d$ in the Hilbert series \cref{eq:hilb-series-det-ideal}
gives the rank of the Macaulay matrix of $\detsystem$ in degree $d$. The
Macaulay matrices computed in \cref{alg:det-F5-corank-one} have full row rank,
allowing for the use of any echelonization algorithm when computing
$\redmacmat_{d,i}$. Hence, the result follows from the complexity of
computing the reduced row echelon form \cite[Sec.\,2.2]{Storjohann2000} (see
also \cite[App.\,A]{JePeSt13}) and the fact that the number of columns in the
Macaulay matrix in degree $2n-3$, the maximal degree of a polynomial in the
grevlex Gröbner basis of $\detideal$, is the number of monomials of degree
$2n-3$ in $\Bbbk[x_1,\dots, x_4]$.
\end{proof}

Asymptotically, the bound given in \cite[Thm.~20]{FaugereSafeySpaenlehauer2013} is in $O\left(n^{5\omega+2}\right)$ whereas that given by \cref{prop:corank-one-complexity} is in $O\left(n^{4\omega-1}\right)$.

\section{Experimental results}
\label{sec:experimental}

Here we present some experimental results on numbers of
reductions to zero in our refinements of the $F_5$ algorithm
compared to the standard $F_5$ algorithm. The systems used for
these results were obtained by building square matrices of
homogeneous linear forms with random coefficients over
$\Bbbk=\FF_{65521}$. This field is large enough that the genericity
assumptions necessary for our results to hold do so with high
probability when taking random coefficients.

All Gröbner basis computations were performed using an
implementation of both standard $F_5$ and our refinements to $F_5$ written in
SageMath (see \cite{SageMath}) using the FFLAS-FFPACK library (see
\cite{fflas-ffpack}) for the linear algebra subroutines. When $r=n-2$, we
compute a full Gröbner basis for $\detsystemCorkOne$, whereas when $r<n-2$, we
only compute a Gröbner basis of $\detsystem$ up to degree $r+2$, as past this
degree our algorithm performs no differently to standard $F_5$.

When $r=n-2$,  all reductions to zero are avoided and thus all Macaulay
matrices are full rank. By virtue of \cref{prop:syz-crit}, if a row of
$\macmat_{d,i}$ reduces to zero, then all multiples of this row in
$\macmat_{d',i}$ for $d'>d$ reduce to zero as well, and the standard $F_5$
algorithm avoids these rows. Note however, that there are a significant number
of reductions to zero which do not arise from reductions to zero in lower
degree, as evidenced by the discrepancy between the size of the generating set
for $\Syz(\detsystemCorkOne)$, which is $2n^2-2$ (when $r=n-2$) and the number
of reductions to zero encountered by the standard $F_5$ algorithm. 

Note also that by \cite[Cor.~19]{FaugereSafeySpaenlehauer2013}, the largest
degree of a polynomial in the reduced grevlex Gröbner basis for
$\detidealCorkOne$ is $2n-3$, which is strictly smaller than $2(r+1)=2n-2$.
Thus, \cref{prop:f5-crit} is never used when running either the standard $F_5$
algorithm, or our refined algorithm on $\detidealCorkOne$.

When $r<n-2$, we avoid all reductions to zero in the Macaulay matrices
$\macmat_{r+2,i}$ for all $1\le i\le\binom{n}{r+1}^2$. As the data in
\cref{table:corank-one} shows, this already allows us to avoid a significant
number of reductions to zero. In fact, in all higher corank cases, over half of
the reductions to zero overall appear in degree $r+2$. The number of reductions
to zero in degree $r+2$ (and thus the size of a minimal generating set for
$\Syz(\detideal)$) appears to be 
\[ 
	\binom{n}{r+2}^2\left(\frac{2(r+2)(r+1)}{n-r-1}+2r+2\right)
.\] 
From this quantity one could derive a refined estimate of the complexity of
\cref{alg:det-F5}. 

Note that generically, in the case $r<n-2$, the largest degree of a polynomial
appearing in the reduced grevlex Gröbner basis for $\detideal$ is
$r\cdot(n-r)+1$ again by \cite[Cor.~19]{FaugereSafeySpaenlehauer2013}. Thus, in
this case, \cref{prop:f5-crit} is used as soon as the degree exceeds $2(r+1)$.

Finally, we observe that the speedups which can already be achieved using
the results of this paper, within our software framework, increase when $n$
grows and $n-r$ remains fixed. In the case where $n-r=2$ we obtain speedup
which are close to $10$. This clearly indicates the potential of these results
with respect to practical computation times.

\begin{table}[htp]
  \caption{\textmd{Reductions to zero in standard \(F_5\) (there is none in determinantal-\(F_5\)) as well as ratio of timings for standard $F_5$ compared to determinantal-$F_5$, when computing a $D$-Gröbner basis for the system of $(r+1)$-minors of a generic $n\times n$ matrix of homogeneous linear forms in $k$ variables over $\field=\FF_{65521}$.}}
\label{table:corank-one}
\centering
\begin{tabular}{cccccc}
	$n$ & $r$ & $k$ & $D$ & \multicolumn{1}{p{1.5cm}}{\centering Red. to $0$ \\ (Std. $F_5$)} & $\frac{\text{(Std. $F_5$)}}{\text{(Det. $F_5$)}}$ \\\hline
4   & 2   & 4   & 5   & 56                 & 0.11                    \\
5   & 3   & 4   & 7   & 129                & 0.08                    \\
6   & 4   & 4   & 9   & 239                & 0.43                    \\
7   & 5   & 4   & 11  & 414                & 0.69                    \\
8   & 6   & 4   & 13  & 663                & 1.46                    \\
9   & 7   & 4   & 15  & 959                & 1.65                    \\
10  & 8   & 4   & 17  & 1387               & 2.26                    \\
11  & 9   & 4   & 19  & 1871               & 3.07                    \\
12  & 10  & 4   & 21  & 2525               & 3.99                    \\
13  & 11  & 4   & 23  & 3181               & 4.94                    \\
14  & 12  & 4   & 25  & 4032               & 6.00                    \\
15  & 13  & 4   & 27  & 4977               & 6.03                    \\
16  & 14  & 4   & 29  & 6213               & 7.93                    \\
17  & 15  & 4   & 31  & 7515               & 7.22                    \\
18  & 16  & 4   & 33  & 8845               & 7.99                    \\
19  & 17  & 4   & 35  & 10544              & 8.65                    \\
20  & 18  & 4   & 37  & 12969              & 10.59                   \\\hline
4   & 1   & 9   & 3   & 160                & 1.27                    \\
5   & 2   & 9   & 4   & 450                & 1.77                    \\
6   & 3   & 9   & 5   & 1008               & 2.04                    \\
7   & 4   & 9   & 6   & 1960               & 2.16                    \\
8   & 5   & 9   & 7   & 3456               & 2.40                    \\
9   & 6   & 9   & 8   & 5670               & 2.50                    \\\hline
5   & 1   & 16  & 3   & 800                & 1.34                    \\
6   & 2   & 16  & 4   & 3150               & 1.59                    \\
7   & 3   & 16  & 5   & 9408               & 1.72                    \\\hline
6   & 1   & 25  & 3   & 2800               & 1.28                    \\
7   & 2   & 25  & 4   & 14700              & 1.39                    \\\hline
7   & 1   & 36  & 3   & 7840               & 1.22                   
\end{tabular}
\end{table}

\section*{Acknowledgements}

The authors are supported by \emph{Quantum
Information Center Sorbonne} (QICS); by the Agence
nationale de la recher\-che (ANR), grant agreements
ANR-19-CE40-0018 \textsc{De Rerum Natura} and
ANR-18-CE33-0011 \textsc{SESAME} projects; by the joint
ANR-Austrian Science Fund FWF
projects ANR-22-CE91-0007 \textsc{EAGLES} and
ANR-FWF ANR-19-CE48-0015 \textsc{ECARP}; and by the
EOARD-AFOSR, grant agreement
FA8665-20-1-7029.

\bibliographystyle{alpha}
\bibliography{DeterminantalF5}

\end{document}